\documentclass[a4paper]{article}
\usepackage{newlfont}
\usepackage{amssymb}
\usepackage{amsfonts}
\usepackage{amsmath}
\usepackage{amsthm}
\usepackage{youngtab}
\usepackage[all]{xy}

\newcommand{\bra}[1]{\ensuremath{\langle #1 |}}
\newcommand{\ket}[1]{\ensuremath{| #1 \rangle}}
\newcommand{\bk}[2]{\ensuremath{\langle #1 | #2 \rangle}}
\newcommand{\kb}[2]{\ensuremath{| #1 \rangle\!\langle #2 |}}

\newtheorem{theo}{Theorem}[section]
\newtheorem{prop}{Proposition}[section]

\newtheorem{cor}{Corollary}[section]

\theoremstyle{definition}
\newtheorem{example}{Example}[section]
\newtheorem{definition}{Definition}[section]
\newtheorem{remark}{Remark}[section]
\newcommand{\we}{\wedge}
\newcommand{\ot}{\otimes}
\newcommand{\ti}{\times}

\newcommand{\cD}{{\cal D}}

\newcommand{\cS}{{\cal S}}

\newcommand{\cP}{{\cal P}}

\newcommand{\cE}{{\cal E}}

\newcommand{\cH}{{\cal H}}

\newcommand{\la}{\langle}
\newcommand{\ra}{\rangle}

\newcommand{\raa}{\rightarrow}

\newcommand{\R}{{\mathbf R}}
\newcommand{\C}{{\mathbb C}}

\newcommand{\wt}{\widetilde}

\mathchardef\za="710B  
\mathchardef\zb="710C  
\mathchardef\zg="710D  
\mathchardef\zd="710E  
\mathchardef\zve="710F 
\mathchardef\zz="7110  
\mathchardef\zh="7111  
\mathchardef\zvy="7112 
\mathchardef\zi="7113  
\mathchardef\zk="7114  
\mathchardef\zl="7115  
\mathchardef\zm="7116  
\mathchardef\zn="7117  
\mathchardef\zx="7118  
\mathchardef\zp="7119  
\mathchardef\zr="711A  
\mathchardef\zs="711B  
\mathchardef\zt="711C  
\mathchardef\zu="711D  
\mathchardef\zvf="711E 
\mathchardef\zq="711F  
\mathchardef\zc="7120  
\mathchardef\zw="7121  
\mathchardef\ze="7122  
\mathchardef\zy="7123  
\mathchardef\zf="7124  
\mathchardef\zvr="7125 
\mathchardef\zvs="7126 
\mathchardef\zf="7127  
\mathchardef\zG="7000  
\mathchardef\zD="7001  
\mathchardef\zY="7002  
\mathchardef\zL="7003  
\mathchardef\zX="7004  
\mathchardef\zP="7005  
\mathchardef\zS="7006  
\mathchardef\zU="7007  
\mathchardef\zF="7008  
\mathchardef\zW="700A  

\newcommand{\be}{\begin{equation}}
\newcommand{\ee}{\end{equation}}

\newcommand{\bea}{\begin{eqnarray}}
\newcommand{\eea}{\end{eqnarray}}
\newcommand{\beas}{\begin{eqnarray*}}
\newcommand{\eeas}{\end{eqnarray*}}

\newcommand{\nm}[1]{\ensuremath{\Vert #1 \Vert}}

\begin{document}

\title{Segre maps and entanglement \\ for multipartite systems of indistinguishable particles}
\author{Janusz Grabowski\footnote{email: jagrab@impan.pl} \\
\textit{Faculty of Mathematics and Natural Sciences, College of Sciences}\\
\textit{Cardinal Stefan Wyszy\'nski University,}\\
\textit{W\'oycickiego 1/3, 01-938 Warszawa, Poland}\\
\\
Marek Ku\'s\footnote{email: marek.kus@cft.edu.pl}\\
\textit{Center for Theoretical Physics, Polish Academy of Sciences,} \\
\textit{Aleja Lotnik{\'o}w 32/46, 02-668 Warszawa,
Poland} \\
\\
Giuseppe Marmo\footnote{email: marmo@na.infn.it}              \\
\textit{Dipartimento di Scienze Fisiche, Universit\`{a} ``Federico II'' di Napoli} \\
\textit{and Istituto Nazionale di Fisica Nucleare, Sezione di Napoli,} \\
\textit{Complesso Universitario di Monte Sant Angelo,} \\
\textit{Via Cintia, I-80126 Napoli, Italy} \\
} \maketitle

\begin{abstract}
We elaborate the concept of entanglement for multipartite system with bosonic and fermionic constituents and its generalization to systems with arbitrary parastatistics. The entanglement is characterized in terms of generalized Segre maps, supplementing thus an algebraic approach to the problem by a more geometric point of view.




\end{abstract}

\section{Introduction}
The possibility of identifying subsystems states in a given total state of a composite quantum system goes under the name of {\it separability}. In the case of pure states such a possibility is guaranteed if the composite state takes the form of the tensor product of subsystems states.

On the other hand, with the advent of Quantum Field Theory, we have identified
elementary particles which are either bosons or fermions. As a matter of fact,
according to the spin-statistics theorem all particles are either bosons or
fermions. The difference is that a state is unchanged by the interchange of two
identical bosons, whereas it changes the sign under the interchange of two
identical fermions. The characterization of fermionic states contains already
the lack of the factorization of the total state of the composite system.
According to usual wisdom, this would always imply the presence of an
entanglement. In our opinion this state of affairs cannot be maintained, so
there is a need of a refinement of the notion of entanglement that describes
better the situation when we are dealing with bosons and fermions or even with
`parabosons' or `parafermions' arising from potentially meaningful
parastatistics \cite{green53,ohnuki82}.

In \cite{grabowski11} we analyzed a concept of entanglement for a multipartite
system with bosonic and fermionic constituents in purely algebraic way using
the the representation theory of the underlying symmetry groups. Correlation
properties of indistinguishable particles become relevant when subsystems are
no longer separated by macroscopic distances, like e.g. in quantum gates based
on quantum dots, where they are confined to the same spatial regions
\cite{schliemann01}. In our approach to bosons and fermions we adopted the
concept of entanglement put forward in \cite{schliemann01,sckll01} for
fermionic systems and extended in \cite{eckert02,paskauskas01} in a natural way
to bosonic ones\footnote{For discussion of a slightly different treatment of
bosons in \cite{li01,ghirardi02,ghirardi04,ghirardi05} see the Introduction and
Conclusion sections of \cite{grabowski11}.}.

The problem of quantifying and measuring entanglement in systems of many
indistinguishable particles remains a topic of vivid investigations from
different points of view. Like in the case of distinguishable particles,
also nonclassical correlations of indistinguishable particles can be
studied with the use of various characteristics,
e.g., a suitably adapted concept of the entropy of a state \cite{buscemi07},
\cite{plastino09} or by mapping the Fock spaces onto spaces of qubits \cite{lari09},
\cite{buscemi11}, elaborating thus a concept of the so called {mode entanglement}
proposed earlier by Zanardi \cite{zanardi02}, \cite{zanardi02a}. Another area
of the current theoretical research on indistinguishable particles concentrates
around problems of identifying and measuring the entanglement in realistic
experimental circumstances \cite{dowling06}, \cite{sasaki11}. An exhaustive
exposition of these aspects of the theory, out of the scope of the present paper, is contained in a
recent review article by Tichy, Mintert, and Buchleitner \cite{tichy11}.

Our approach appeared to be sufficiently general to define entanglement also
for systems with an arbitrary {\it parastatistics} in a consistent and unified
way. For pure states we defined the S-rank, generalizing the notion of the
Schmidt rank for distinguishable particles and playing an analogous role in the
characterization of the degree of entanglement among particles with arbitrary
exchange symmetry (parastatistics).

In the algebraic geometry, a canonical embedding of the product $\C
P^{n-1}\ti\C P^{m-1}$ of complex projective spaces into $\C P^{nm-1}$ is known
under the name the {\it Segre embedding} (or the {\it Segre map}). In the
quantum mechanical context, the complex projective space $\C P^{n-1}$
represents pure states in the Hilbert space $\C^n$, and $\C P^{nm-1}$ represent
pure states in $\C^n\otimes\C^m$, so that the Segre embedding gives us a
geometrical description of separable pure states and, as shown in
\cite{grabowski05,grabowski06}, this description can be extended also to mixed
states.

In the present paper we give a geometric description of the entanglement for
systems with arbitrary symmetry (with respect to exchanging of subsystems) in
terms of {\it generalized Segre embeddings} associated with particular
parastatistics. This description is complementary to the one presented in
\cite{grabowski11} in terms of the S-rank. For systems with arbitrary exchange
symmetries, unlike for the systems of distinguishable particles, the spaces of
states are not, in general, projectivizations of the full tensor products of
the underlying Hilbert spaces of subsystems, but rather some parts of them. We
show in the following how to extend properly the concept of the Segre embedding
to achieve a geometric description analogous to that for distinguishable
particles. This approach uses a unifying mathematical framework based on the
representation theory and strongly suggesting certain concepts of the
separability, thus of the entanglement, in the case of indistinguishable
particles. For physicists, this approach may be viewed as being too
mathematical and abstract, but in our opinion it covers exactly the logic
structure of the notion of entanglement for systems of particles with some
symmetries.

In the next section we shortly review the relevant concepts of composite
systems of distinguishable particles, their description in terms of the
classical Segre maps, as well as entanglement measures for systems of
distinguishable particles. In sections 3 and 4 we give a brief review of the
algebraic description of the entanglement for bosons and fermions in terms of
the S-rank of tensors presented in \cite{grabowski11}. The main results are
contained in Sections~5-7 where we construct the Segre maps, first for bosons
and fermions and, finally, for systems with an arbitrary parastatistics. The
constructions are based on Theorem \ref{crucial} which relates simple tensor of
a given parastatistics with the corresponding Young diagram.

\section{Composite systems, separability, and entanglement}
Let $\cH$ be a Hilbert space with a Hermitian product $\bk{\cdot}{\cdot}$,
$gl(\cH)$ be the vector space of complex linear operators on $\cH$, $GL(\cH)$
be the group of invertible operators from $gl(\cH)$, and $U(\cH)$ be its
subgroup of unitary operators on $\cH$. For simplicity, we will assume that
$\cH$ is finite-dimensional, say $\dim(\cH)=n$, but a major part of our work
remains valid also for Hilbert spaces of infinite dimensions. Note only that in
the infinite dimensions the corresponding tensor product
$\mathcal{H}_1\otimes\mathcal{H}_2$ is the tensor product in the category of
Hilbert spaces, i.e., corresponding to the Hilbert-Schmidt norm.

With $u(\cH)$ we will denote the Lie algebra of the Lie group $U(\cH)$
consisting of anti-Hermitian operators, while $u^*(\cH)$ will denote its dual
interpreted as the Euclidean space of Hermitian operators with the scalar
product \be\label{metric1}\la A,B\ra_{u^*}=\frac{1}{2}\text{Tr}(AB). \ee The
space of non-negatively defined operators from $gl(\cH)$, i.e. of those $\zr\in
gl(\cH)$ which can be written in the form $\zr=T^\dag T$ for a certain $T\in
gl(\cH)$, we denote as $\cP(\cH)$. It is a convex cone in $gl(\cH)$ and the set
of {\it density states} $\cD(\cH)$ is distinguished in $\cP(\cH)$ by the
normalizing condition $\text{Tr}(\zr)=1$. We will regard $\cP(\cH)$ and
$\cD(\cH)$ as embedded in $u^*(\cH)$, so that the space $\cD(\cH)$ of density
states is a convex set in the affine hyperplane in $u^*(\cH)$ determined by the
equation $\text{Tr}(\rho)=1$. As the difference of two Hermitian operators of
trace 1 (the vector connection two points in the affine hyperplane) is a
Hermitian operator of trace 0, the model vector spaces of this affine
hyperplane is therefore canonically identified with the space of Hermitian
operators with the trace $0$.

Denote the set of all operators from $\cD(\cH)$ of rank $k$ with $\cD^k(\cH)$.
It is well known that the set of extreme points of $\cD(\cH)$ coincides with
the set $\cD^1(\cH)$ of {\it pure states}, i.e. the set of one-dimensional
orthogonal projectors $\rho_x=\mid x\ra\la x\mid$, $\Vert x\Vert=1$. Hence,
every element of $\cD(\cH)$ is a convex combination of points from
$\cD^1(\cH)$. The space $\cD^1(\cH)$ of all pure states can be identified with
the complex projective space $\mathbb{P}\cH\simeq\C P^{n-1}$ {\it via} the
projection

\be\label{pure}\cH\setminus\{ 0\}\ni x\mapsto \zr_x=\frac{\mid\! x\ra\la
x\!\!\mid}{\Vert x\Vert^2}\in\cD^1(\cH) \ee
which identifies the points of the
orbits of the action in $\cH$ of the multiplicative group $\C\setminus\{ 0\}$
by complex homotheties. It is well known that the complex projective space
$\mathbb{P}\cH=\cD^1(\cH)$ is canonically a K\"ahler manifold. The symplectic
structure on $\cD^1(\cH)\subset u^*(\cH)$ is the canonical symplectic structure
of an $U(\cH)$-coadjoint orbit, and the metric, called the {\it Fubini-Study
metric} is just the metric induced from the embedding of $\cD^1(\cH)$ into the
Euclidean space $u^*(\cH)$. This is the best known compact K\"ahler manifold in
the algebraic geometry.

Suppose now that our Hilbert space has a fixed decomposition into the tensor
product of two Hilbert spaces, $\cH=\cH^1\ot\cH^2$. This additional input is
crucial in studying composite quantum systems. Observe first that the tensor
product map \be\otimes:\cH^1\ti\cH^2\raa\cH=\cH^1\ot\cH^2 \ee associates the
product of rays with a ray, so it induces a canonical embedding on the level of
complex projective spaces,
\bea \text{Seg}:\mathbb{P}\cH^1\ti \mathbb{P}\cH^2&\raa& \mathbb{P}\cH=\mathbb{P}(\cH^1\ot\cH^2),\\
(\mid\!x^1\ra\la x^1\!\!\mid,\mid\!x^2\ra\la x^2\!\!\mid) &\mapsto&
\mid\!x^1\ot x^2\ra\la x^1\ot x^2\!\!\mid\,,\quad \nm{x^1}=\nm{x^2}=1\,.
\eea
This embedding of the product of complex projective spaces into the projective
space of the tensor product is called in the literature the {\it Segre
embedding} \cite{segre}. Note that the elements of the image
$\text{Seg}(\mathbb{P}\cH^1\ti \mathbb{P}\cH^2)$ in
$\mathbb{P}\cH=\mathbb{P}(\cH^1\ot\cH^2)$ are usually called {\it separable
pure states} (separable with respect to the decomposition $\cH=\cH^1\ot\cH^2$).

The Segre embedding is related to the (external) tensor product of the basic
representations of the unitary groups $U(\cH^1)$ and $U(\cH^2)$, i.e. with the
representation of the direct product group in $\cH=\cH^1\ot\cH^2$, \beas
U(\cH^1)\ti U(\cH^2)\ni(T_1,T_2)&\mapsto& T_1\ot T_2\in
U(\cH)=U(\cH^1\ot\cH^2),\\
(T_1\ot T_2)(x^1\ot x^2)&=&T_1(x^1)\ot T_2(x^2).
\eeas
Note that $T_1\ot T_2$ is unitary, since the Hermitian product in $\cH$ is
related to the Hermitian products in $\cH^1$ and $\cH^2$ by \be\label{sp1} \la
x^1\ot x^2,y^1\ot y^2\ra_{\cH}=\la x^1,y^1\ra_{\cH^1}\cdot \la
x^2,y^2\ra_{\cH^2}. \ee The above group embedding gives rise to the
corresponding embedding of Lie algebras or, by our identification, of their
duals which, with some abuse of notation, we will denote by
\be\label{imb}\text{Seg}:u^*(\cH^1)\ti u^*(\cH^2)\raa u^*(\cH),\quad
(A,B)\mapsto A\ot B. \ee The original Segre embedding is just the latter map
reduced to pure states. In fact, a stronger result holds true
\cite{grabowski05,grabowski06}.
\begin{prop} The embedding (\ref{imb}) maps
$\cD^k(\cH^1)\ti\cD^l(\cH^2)$ into $\cD^{kl}(\cH)$.
\end{prop}

\medskip\noindent
Let us denote the image $\text{Seg}(\cD^1(\cH^1)\ti\cD^1(\cH^2))$, i.e.  the set of {\it separable pure states},
with $\cS^1(\cH)$, and its convex hull, ${conv}\left(\cS^1(\cH)\right)$, i.e. the set of all {\it mixed separable
states} in $u^*(\cH)$, with $\cS(\cH)$. The states from
$$\cE(\cH)=\cD(\cH)\setminus\cS(\cH),$$
i.e. those which are not separable, are called {\it entangled states}. It is well known (see e.g. \cite{grabowski05})
that $\cS^1(\cH)$ is exactly the set of extremal points of $\cS(\cH)$. What we have just presented is a very simple geometric interpretation of separability and entanglement.
\medskip\noindent

The entangled states play an important role in quantum computing and one of main problems is to decide
effectively whether a given composite state is entangled or not. An abstract measurement of entanglement can
be based on the following observation (see also \cite{Vidal})

\smallskip
Let $E$ be the set of all extreme points of a compact convex set $K$ in a finite-dimensional real vector space
$V$, and let $E_0$ be a compact subset of $E$ with the convex hull $K_0=\text{conv}(E_0)\subset K$. For every
non-negative function $f:E\raa\R_+$  define its extension ({\it convex roof}) $f_K:K\raa\R_+$ by
\be\label{cf}f_K(x)=\inf_{x=\sum t_i\za_i}\sum t_if(\za_i),
\ee
where the {\it infimum} is taken with respect to all expressions of $x$ in the form of convex combinations of
points from $E$. Recall that, according to Krein-Milman theorem, $K$ is the convex hull of its extreme
points.
\begin{prop}\label{hullt}
For every non-negative continuous function $f:E\raa\R_+$ which vanishes exactly on $E_0$, the function $f_K$ is
convex on $K$ and vanishes exactly on $K_0$
\end{prop}
An immediate consequence is the following (cf. \cite{grabowski05,grabowski06}).
\begin{cor}
Let $\cH=\cH^1\ot\cH^2$ and let $F:\cD^1(\cH)\raa\R_+$ be a continuous function which vanishes exactly on on the set
$\cS^1(\cH)$ of separable pure states. Then,
$$\zm=F_{\cD(\cH)}:\cD(\cH)\raa\R_+$$
is a measure of entanglement, i.e. $\zm$ is convex and $\zm(\zr)=0$ if and only if the (mixed) density state $\zr$ is separable. Moreover, if
$f$ is taken $U(\cH^1)\ti U(\cH^2)$-invariant, then $\zm$ is
$U(\cH^1)\ti U(\cH^2)$-invariant.
\end{cor}
\begin{remark}
In the terminology of \cite{Vidal}, the convex roof function $F_{\cD(\cH)}$ is {\it entanglement monotone} if $F$ is entanglement monotone on pure states.
\end{remark}

\section{Tensor algebras, fermions, and bosons}
\label{sec:tensoralgebras}

To describe some properties of systems composed of
indistinguishable particles and to fix the notation, let us start
with introducing corresponding tensor algebras associated with a
Hilbert space $\mathcal{H}$.

In the tensor power $\mathcal{H}^{\otimes
k}=\underset{k-\mathrm{times}}{\underbrace{\mathcal{H}\otimes\cdots\otimes\mathcal{H}}}$,
we distinguish the subspaces: $\mathcal{H}^{\vee
k}=\underset{k-\mathrm{times}}{\underbrace{\mathcal{H}\vee\cdots\vee\mathcal{H}}}$
of totally symmetric tensors and $\mathcal{H}^{\wedge
k}=\underset{k-\mathrm{times}}{\underbrace{\mathcal{H}\wedge\cdots\wedge\mathcal{H}}}$
of totally antisymmetric ones, together with the symmetrization,
$\pi_k^\vee:\mathcal{H}^{\otimes k}\rightarrow\mathcal{H}^{\vee
k}$, and antisymmetrization, $\pi_k^\wedge:\mathcal{H}^{\otimes
k}\rightarrow\mathcal{H}^{\wedge k}$, projectors:
\begin{eqnarray}\label{sproj}
\pi_k^\vee(f_1\otimes\cdots\otimes f_k)&=&
\frac{1}{k!}\sum_{\sigma\in S_k}f_{\sigma(1)}\otimes\cdots\otimes
f_{\sigma(k)},
\\
\pi_k^\wedge(f_1\otimes\cdots\otimes f_k)&=&
\frac{1}{k!}\sum_{\sigma\in S_k}(-1)^\sigma
f_{\sigma(1)}\otimes\cdots\otimes f_{\sigma(k)}.
\end{eqnarray}
Here, $S_k$ is the group of all permutations
$\sigma:\{1,\ldots,k\}\rightarrow\{1,\ldots,k\}$, and
$(-1)^\sigma$ denotes the sign of the permutation $\sigma$. Note
that with every permutation $\sigma\in S_k$ there is a canonically
associated  unitary operator $U_\sigma$ on $\mathcal{H}^{\otimes
k}$ defined by $$U_\sigma(f_1\otimes\cdots\otimes
f_k)=f_{\sigma(1)}\otimes\cdots\otimes f_{\sigma(k)}\,,$$ so that
the map $\sigma\mapsto U_\sigma$ is an injective unitary
representation of $S_k$ in $\mathcal{H}^{\otimes k}$. We
will write simply $\sigma$ instead of $U_\sigma$ if no
misunderstanding is possible. Symmetric and skew-symmetric tensors
are characterized in terms of this unitary action by $\sigma(v)=v$
and $\sigma(v)=(-1)^\sigma v$, respectively, for all $\sigma\in
S_k$

We put, by convention, $\mathcal{H}^{\otimes 0}=\mathcal{H}^{\vee
0}=\mathcal{H}^{\wedge 0}=\mathbb{C}$. It is well known that the
obvious structure of a unital graded associative algebra on the
graded space $\mathcal{H}^\otimes
=\mathop\otimes\limits_{k=0}^\infty\mathcal{H}^{\otimes k}$ (the
{\it tensor algebra}) induces canonical unital graded associative algebra structures
on the spaces $\mathcal{H}^\vee
=\mathop\oplus\limits_{k=0}^\infty\mathcal{H}^{\vee k}$ (called
the {\it bosonic Fock space}) and $\mathcal{H}^\wedge
=\mathop\oplus\limits_{k=0}^\infty\mathcal{H}^{\wedge k}$ (called
the {\it fermionic Fock space}) of symmetric and antisymmetric
tensors. This simply means that we have associative multiplications \begin{equation}\label{vee}
v_1\vee v_2=\pi^\vee(v_1\otimes v_2),
\end{equation}
\begin{equation}\label{wedge}
w_1\wedge w_2=\pi^\wedge(w_1\otimes w_2),,
\end{equation}
where
\begin{equation}\label{sym}
\pi^\vee=\mathop\oplus\limits_{k=0}^\infty
\pi_k^\vee:\mathcal{H}^\otimes\rightarrow\mathcal{H}^\vee,
\end{equation}
and
\begin{equation}\label{asym}
\pi^\wedge=\mathop\oplus\limits_{k=0}^\infty
\pi_k^\wedge:\mathcal{H}^\otimes\rightarrow\mathcal{H}^\wedge,
\end{equation}
are the symmetrization and antisymmetrization projections. Moreover, these multiplications respect the grading, i.e. $v_1\vee v_2\in \cH^{\vee (k+l)}$ if $v_1\in\cH^{\vee k}$, $v_2\in\cH^{\vee l}$ , and $w_1\we w_2\in \cH^{\we (k+l)}$ if $w_1\in\cH^{\we k}$, $w_2\in\cH^{\we l}$.
Note also
that the multiplication in $\mathcal{H}^\vee$ is commutative,
$v_1\vee v_2= v_2\vee v_1$, and the multiplication in
$\mathcal{H}^\wedge$ is graded commutative, $w_1\wedge
w_2=(-1)^{k_1\cdot k_2} w_2\wedge w_1$, for
$w_i\in\mathcal{H}^{\wedge k_i}$.

It is well known that the symmetric tensor algebra $\mathcal{H}^\vee$ can be canonically
identified with the algebra $Pol(\mathcal{H})$ of polynomial
functions on $\mathcal{H}$.
Indeed, any $f\in\mathcal{H}$ can be identified with the linear
function $x_f$ on $\mathcal{H}$ by means of the Hermitian product: $x_f(y)=\langle f| y\rangle$. We must stress, however, that the identification $f\mapsto x_f$ is anti-linear.
This can be extended to an anti-linear isomorphism of commutative algebras in
which $f_1\vee\cdots\vee f_k$ corresponds to the homogenous
polynomial $x_{f_1}\cdots x_{f_k}$. Similarly, one identifies
$\mathcal{H}^\wedge$ with the Grassmann algebra $Grass(\cH)$ of
polynomial (super)functions on $\mathcal{H}$. Here, however,
with $f\in\mathcal{H}$ we associate a linear function $\xi_f$ on
$\mathcal{H}$ regarded as and odd function:
$\xi_f\xi_{f^\prime}=-\xi_{f^\prime}\xi_f$. In the language of
supergeometry one speaks about the purely odd vector space
$\Pi\mathcal{H}$ obtained from $\mathcal{H}$ by changing the parity.

If we fix a basis $e_1,\ldots,e_n$ in $\mathcal{H}$ and associate
with its elements even linear functions $x_1,\ldots,x_n$ on
$\mathcal{H}$, and odd linear functions $\xi_1,\ldots,\xi_n$ on
$\Pi\mathcal{H}$, then
$\mathcal{H}^\vee$ can be identified with the algebra of complex
polynomials in $n$ commuting variables, $\mathcal{H}^\vee\simeq\mathbb{C}[x_1,\ldots,x_n]$. Similarly,
$\mathcal{H}^\wedge\simeq\mathbb{C}[\xi_1,\ldots,\xi_n]$, i.e.,
$\mathcal{H}^\wedge$ can be identified with the algebra of complex
Grassmann polynomials in $n$ anticommuting variables. The
subspaces $\mathcal{H}^{\vee k}$ and $\mathcal{H}^{\wedge k}$
correspond to homogenous polynomials of degree $k$. It is
straightforward that homogeneous polynomials $x_1^{k_1}\cdots
x_n^{k_n}$, with $k_1+\cdots+k_n=k$, form a basis of
$\mathcal{H}^{\vee k}$, while homogeneous Grassmann polynomials
$\xi_{i_1}\wedge\cdots\wedge\xi_{i_k}$, with $1\le
i_1<i_2<\cdots<i_k\le n$, form a basis of $\mathcal{H}^{\wedge
k}$. In consequence, $\dim\mathcal{H}^{\vee
k}={{n+k-1}\choose{k}}$ and $\dim\mathcal{H}^{\wedge
k}={{n}\choose{k}}$, so the gradation in the fermionic Fock space
is finite-dimensional (for a finite-dimensional $\mathcal{H})$.

\medskip
Note that any basis $\{e_1,\ldots, e_n\}$ in $\mathcal{H}$ induces
a basis $\big\{e_{i_1}\otimes e_{i_2}\otimes\cdots\otimes
e_{i_k}\, |\; i_1,\ldots, i_k\in\{1,\ldots,n\}\big\}$ in
$\mathcal{H}^{\otimes k}$. Therefore, any
$u\in\mathcal{H}^{\otimes k}$ can be uniquely written as a linear
combination
\begin{equation}\label{gentens}
u=\sum_{i_1,\ldots,i_k=1}^n u^{i_1\ldots i_k}e_{i_1}
\otimes\ldots\otimes e_{i_k}.
\end{equation}
If $u\in\mathcal{H}^{\vee k}$, then the tensor coefficients
$u^{i_1\ldots i_k}$ are totally symmetric and, after applying the
symmetrization projection to (\ref{gentens}), we get
\begin{equation}\label{symtens}
u=\sum_{i_1,\ldots,i_k=1}^n u^{i_1\ldots i_k}e_{i_1}
\vee\ldots\vee e_{i_k}.
\end{equation}
Similarly, if $u\in\mathcal{H}^{\wedge k}$, the tensor
coefficients $u^{i_1\ldots i_k}$ are totally antisymmetric and
\begin{equation}\label{antisymtens}
u=\sum_{i_1,\ldots,i_k=1}^n u^{i_1\ldots i_k}e_{i_1}
\wedge\ldots\wedge e_{i_k}.
\end{equation}
We will refer to the coefficients $u^{i_1\ldots i_k}$ as to the
{\it coefficients of $u$ in the basis} $\{e_1,\ldots, e_n\}$.

The Hermitian product in $\mathcal{H}$ has an obvious extension to a Hermitian product in
$\mathcal{H}^{\otimes k}$,
\begin{equation}\label{kpairing}
\langle f_1\otimes\cdots\otimes f_k|  g_1\otimes\cdots\otimes
g_k\rangle =\prod_{i=1}^k\langle f_i| g_i\rangle,
\end{equation}
and viewing symmetric and antisymmetric tensors as canonically
embedded in the tensor algebra, we find the corresponding Hermitian products in
$\mathcal{H}^{\vee k}$ and $\mathcal{H}^{\wedge k}$.

For $f_1,\dots,f_k\in\cH$ and $g_1,\dots,g_k\in\cH$, we get
\begin{equation}\label{pairingvee}
\langle f_1\vee\cdots\vee f_k|  g_1\vee\cdots\vee g_k\rangle
=\frac{1}{(k!)^2}\sum_{\sigma,\tau\in S_k}\prod_{i=1}^k \langle
f_{\sigma(i)}| g_{\tau(i)}\rangle =\frac{1}{k!}\mathrm{per}
(\langle f_i| g_j\rangle).
\end{equation}
Here, $\frac{1}{k!}\sum_{\tau\in S_k}\prod_{i=1}^k
a_{i\tau(i)}=\mathrm{per}(a_{ij})$ is the permanent of the matrix
$A=(a_{ij})$. Similarly,
\begin{equation}\label{pairingwedge}
\langle f_1\wedge\cdots\wedge f_k|  g_1\wedge\cdots\wedge
g_k\rangle =\frac{1}{k!}\det(\langle f_i| g_j\rangle).
\end{equation}

These Hermitian products can be
generalized to certain `pairings' ({\it contractions} or {\it
inner products}) between $\mathcal{H}^{\vee k}$ and
$\mathcal{H}^{\vee l}$ on one hand, and $\mathcal{H}^{\wedge
k}$ and $\mathcal{H}^{\wedge l}$ on the other, $l\le k$. For the
standard simple tensors $f=f_1\otimes\cdots\otimes f_k\in
\mathcal{H}^{\otimes k}$ and $ g= g_1\otimes\cdots\otimes g_l\in
\mathcal{H}^{\otimes l}$, we just put
$$\imath_gf=\langle
f_1\otimes\cdots\otimes f_l|  g_1\otimes\cdots\otimes g_l\rangle
f_{l+1}\otimes\cdots\otimes f_k$$ and extend it by linearity to
all tensors. It is easy to see now that, if $v=f_1\vee\cdots\vee
f_k\in \mathcal{H}^{\vee k}\subset\mathcal{H}^{\otimes k}$ and
$\nu=g_1\vee\cdots\vee g_l\in \mathcal{H}^{\vee l}
\subset\mathcal{H}^{\otimes l}$, then
$\imath_{\nu}v\in\mathcal{H}^{\vee (k-l)}$.

Similarly, $\imath_{\omega}w\in\mathcal{H}^{\wedge (k-l)}$, if
$w\in \mathcal{H}^{\wedge k}\subset\mathcal{H}^{\otimes k}$ and
$\omega\in \mathcal{H}^{\wedge l}
\subset\mathcal{H}^{\otimes l}$. Explicitly,
\begin{eqnarray}
\imath_{g_1\vee\cdots\vee g_l}f_1\vee\cdots\vee f_k
=\frac{1}{k!\,l!} \sum_{\substack{\sigma\in S_k\\ \tau\in
S_l}}\prod_{j=1}^l \langle f_{\sigma(j)}|g_{\tau(j)}\rangle
f_{\sigma(l+1)}
\otimes\cdots\otimes f_{\sigma(k)} \nonumber \\
=\frac{(k-l)!}{k!} \sum_{\substack{S\in S(l,k-l)\\ \tau\in
S_l}}\prod_{j=1}^l \langle f_{S(j)}|g_{\tau(j)}\rangle f_{S(l+1)}
\vee\cdots\vee f_{S(k)},
\end{eqnarray}
where $S(l,k-l)$ denotes the group of all $(l,k-l)$ shuffles.
Recall that a permutation $\tau$ in $S_{p + q}$ is a {\it $(p,q)$ shuffle} if
$\zt(1)<\cdots<\tau(p)$ and $\zt(p+1)<\cdots<\tau(p+q)$.

For skew-symmetric tensors,
\begin{eqnarray}
\imath_{g_1\wedge\cdots\wedge g_l}f_1\wedge\cdots\wedge f_k
=\frac{1}{k!\,l!} \sum_{\substack{\sigma\in S_k\\ \tau\in
S_l}}(-1)^\sigma(-1)^\tau \prod_{j=1}^l\langle
f_{\sigma(j)}|g_{\tau(j)}\rangle f_{\sigma(l+1)}
\otimes\cdots\otimes f_{\sigma(k)} \nonumber \\
=\frac{(k-l)!}{k!} \sum_{\substack{S\in S(l,k-l)\\ \tau\in
S_l}}(-1)^\sigma (-1)^\tau \prod_{j=1}^l\langle
f_{S(j)}|g_{\tau(j)}\rangle f_{S(l+1)} \wedge\cdots\wedge
f_{S(k)}.
\end{eqnarray}
In particular,
\begin{equation}
\imath_{g_1\vee\cdots\vee g_k}f_1\vee\cdots\vee f_k = \langle
f_1\vee\cdots\vee f_k| g_1\vee\cdots\vee g_k \rangle,
\end{equation}
and
\begin{equation}
\imath_{g_1\wedge\cdots\wedge g_k}f_1\wedge\cdots\wedge f_k =
\langle f_1\wedge\cdots\wedge f_k| g_1\wedge\cdots\wedge g_k
\rangle.
\end{equation}
Moreover,
\begin{equation}
\imath_{g_1\vee\cdots\vee g_{k-1}}f_1\vee\cdots\vee f_k =
\frac{1}{k!}\sum_{j=1}^k \langle
f_1\vee\overset{\underset{\vee}{j}}{\cdots}\vee f_k|
g_1\vee\cdots\vee g_{k-1} \rangle f_j\,,
\end{equation}
and
\begin{equation}
\imath_{g_1\wedge\cdots\wedge g_{k-1}}f_1\wedge\cdots\wedge f_k =
\frac{1}{k!}\sum_{j=1}^k(-1)^{k-j} \langle
f_1\wedge\overset{\underset{\vee}{j}}{\cdots}\wedge f_k|
g_1\wedge\cdots\wedge g_{k-1} \rangle f_j,,
\end{equation}
where $\overset{\underset{\vee}{j}}{}$ stands for the omission.

\section{The S-rank and entanglement for multipartite Bose and Fermi systems}
\label{sec:c-rank}

There are many concepts of a rank of a tensor used in describing its
complexity. One of the simplest and most natural is the one based
on the inner product operators defined in the previous section. This rank, called in \cite{grabowski11} the \textit{S-rank} and used there to define the entanglement for systems of indistinguishable particles, is a
natural generalization of the {\it Schmidt rank} of 2-tensors.
\begin{definition}\label{def:c-rank}
Let $u\in\mathcal{H}^{\otimes k}$. By the {\it S-rank} of $u$, we
understand the maximum of dimensions
of the linear spaces $\imath_{\mathcal{H}}^{k-1}\sigma(u)$, for
$\sigma\in S_k$, which are the images of the contraction maps
\begin{equation}\label{1}
\mathcal{H}^{\otimes (k-1)}\ni \nu\mapsto
\imath_{\nu}\sigma(u)\in\mathcal{H}.
\end{equation}
Non-zero tensors of minimal S-rank in $\mathcal{H}^{\otimes k}$ (resp.,
$\mathcal{H}^{\vee k}$, $\mathcal{H}^{\wedge k}$) we will call
{\it simple} (resp., {\it simple symmetric, simple antisymmetric}).
\end{definition}

Note that he above definition has its natural counterpart for distinguishable particles, so tensors from $\cH_1\ot\cdots\ot\cH_k$.
We just do the contractions with tensors from $\cH_1\ot\cdots\ot\cH_{k-1}$ and the corresponding permutations.
If particles are identical, $\cH_i=\cH$, and indistinguishable, e.g. the tensors are symmetric or skew-symmetric, we can skip using permutations. In other words, for $u\in\mathcal{H}^{\vee k}$ (resp., $u\in\mathcal{H}^{\wedge
k}$), the S-rank of $u$ equals the dimension of
the linear space which is the image of the contraction map,
\begin{equation}\label{2}
\mathcal{H}^{\vee (k-1)}\ni \nu\mapsto
\imath_{\nu}u\in\mathcal{H},
\end{equation}
(resp.,
\begin{equation}\label{3}
\mathcal{H}^{\wedge(k-1)}\ni \nu\mapsto
\imath_{\nu}u\in\mathcal{H}).
\end{equation}

\begin{theo}\label{theo:simple} (\cite{grabowski11}) \
\begin{description}
\item{(a)} The minimal possible S-rank of a non-zero tensor
    $u\in\mathcal{H}^{\otimes k}$ equals $1$. A tensor
    $u\in\mathcal{H}^{\otimes k}$ is of S-rank $1$ if and only if $u$
    is decomposable, i.e., it can be written in the form
\begin{equation}\label{rank1}
u=f_1\otimes\cdots\otimes f_k, \quad f_i\in \mathcal{H},\quad
f_i\ne 0.
\end{equation}
Such tensors span $\mathcal{H}^{\otimes k}$. \item{(b)} The
minimal possible S-rank of a non-zero tensor
    $v\in\mathcal{H}^{\vee k}$ equals $1$. A tensor
    $v\in\mathcal{H}^{\vee k}$ is of S-rank $1$ if and only if $v$ can
    be written in the form
\begin{equation}\label{rank1a}
v=f\vee\cdots\vee f, \quad f\in \mathcal{H},\quad f\ne 0.
\end{equation}
Such tensors span $\mathcal{H}^{\vee k}$. \item{(c)} The minimal
possible S-rank of a non-zero tensor
    $w\in\mathcal{H}^{\wedge k}$ equals $k$. A tensor
    $w\in\mathcal{H}^{\wedge k}$ is of S-rank $k$ if and only if $w$
    can be written in the form
\begin{equation}\label{rank1b}
w=f_1\wedge\cdots\wedge f_k,
\end{equation}
where $f_1,\ldots,f_k\in\mathcal{H}$ are linearly independent.
Such tensors span $\mathcal{H}^{\wedge k}$.
\end{description}
In particular, the S-rank is 1 for simple and simple
symmetric tensors and it is $k$ for simple antisymmetric tensors
from $\cH^{\we k}$. Simple tensors have the form (\ref{rank1}),
simple symmetric tensors have the form (\ref{rank1a}), and simple
antisymmetric tensors have the form (\ref{rank1b}).
\end{theo}

Using the concept of simple tensors we can define simple
(non-entangled or separable) and entangled pure states for
multipartite systems of bosons and fermions.
\begin{definition}\
\begin{description}

\item{(a)} A pure state $\rho_x$ on $\mathcal{H}^{\vee k}$ (resp.,
on $\mathcal{H}^{\we k}$), $\rho_x=\frac{\kb{x}{x}}{||x||^2}$,
with $x\in\mathcal{H}^{\vee k}$ (resp., $x\in\mathcal{H}^{\we
k}$), $x\ne 0$, is called a \emph{bosonic} (resp.,
\emph{fermionic}) {\it simple} (or {\it non-entangled}) {\it pure
state} if $x$ is a simple symmetric (resp., antisymmetric)
tensor. If $x$ is not simple symmetric (resp., antisymmetric), we
call $\rho_x$ a \emph{bosonic} (resp., \emph{fermionic})
\emph{entangled state}.

\item{(b)} A mixed state $\rho$ on $\mathcal{H}^{\vee k}$ (resp., on
$\mathcal{H}^{\we k}$) we call \emph{bosonic} (resp.,
\emph{fermionic}) {\it simple} (or {\it non-entangled}) {\it mixed
state} if it can be written as a convex combination of {bosonic}
(resp., {fermionic}) simple pure states. In the other case, $\rho$
is called \emph{bosonic} (resp., \emph{fermionic}) \emph{entangled
mixed state}.
\end{description}
\end{definition}

According to Theorem \ref{theo:simple}, bosonic simple pure
$k$-states are of the form
$$\kb{e{\vee}\cdots\vee e}{e{\vee}\cdots\vee e}$$
for unit vectors $e\in\cH$, and fermionic simple pure
$k$-states are of the form
$${k!}\kb{e_1\we\cdots\we e_k}{e_1\we\cdots\we e_k}$$ for
orthonormal systems $e_1,\dots,e_k$ in $\cH$.

Fixing a base in $\mathcal{H}$ results in defining coefficients
$[u^{i_1\ldots i_k}]$ of $u\in\mathcal{H}^{\otimes k}$. Formulae
characterizing simple tensors, thus simple pure states, can be
written in terms of quadratic equations with respect to these
coefficients as follows. The corresponding characterization of
entangled pure states are obtained by negation of the latter.

\begin{theo}\label{theo:quadratic1}(\cite{grabowski11}) \
\begin{description}

\item{(a)} The pure state $\rho_u$,
associated with a tensor $u=[u^{i_1\ldots
i_k}]\in\mathcal{H}^{\otimes k}$, is entangled if and only if
there exist $i_1,\ldots,i_k,j_1,\ldots,j_k$, and  $s=1,\ldots,k$
such that
\begin{equation}\label{quadratic}
u^{i_1\ldots i_s\ldots i_k}u^{j_1\ldots j_s\ldots j_k}\ne
u^{i_1\ldots j_s\ldots i_k}u^{j_1\ldots i_s\ldots j_k}.
\end{equation}

\item{(b)} The bosonic pure state $\rho_v$, associated with a symmetric
tensor $v=[v^{i_1\ldots i_k}]\in\mathcal{H}^{\vee k}$, is bosonic
entangled if and only if there exist $i_1,\ldots,i_k$,
$j_1,\ldots,j_k$, such that
\begin{equation}\label{quadratic2}
v^{i_1\ldots i_{k-1}i_k}v^{j_1\ldots j_{k-1}j_k}\ne v^{i_1\ldots
i_{k-1}j_k}v^{j_1\ldots j_{k-1}i_k}\,.
\end{equation}

\item{(c)} The fermionic pure state $\rho_w$, associated with an
antisymmetric tensor $w=[w^{i_1\ldots i_k}]\in\mathcal{H}^{\wedge
k}$, is fermionic entangled if and only if there exist
$i_1,\ldots,i_{k+1},j_1,\ldots,j_{k-1}$ such that
\begin{equation}\label{quadratic3}
w^{[i_1\ldots i_{k}}w^{i_{k+1}]j_1\ldots j_{k-1}}\ne 0\,,
\end{equation}
where the left-hand side is the antisymmetrization of $w^{i_1\ldots i_{k}}w^{i_{k+1}j_1\ldots j_{k-1}}$ with respect to the indices $i_1,\dots,i_{k+1}$.
\end{description}
\end{theo}
Note that the opposite to (\ref{quadratic3}), $w^{[i_1\ldots i_{k}}w^{i_{k+1}]j_1\ldots j_{k-1}}= 0$, are sometimes called the {\it Pl\"ucker relations}.

\begin{example} Assume that deal with qubit systems, and $\ket{0},\ket{1}$ is an orthonormal basis in $\cH$. The tensor $u=\ket{0}\ot \ket{0}$ has the S-rank 1:
$$\imath_{a\ket{0}+b\ket{1}}\zs(u)=\imath_{a\ket{0}+b\ket{1}}(u)=(a\bra{0}+b\bra{1})\ket{0}\ket{0}=a\ket{0}$$
which is the 1-dimensional space spanned by $\ket{0}$,
while the tensor $u_{\pm}=\ket{0}\ot \ket{1}\pm \ket{1}\ot \ket{0}$ has the S-rank 2:
\beas\imath_{a\ket{0}+b\ket{1}}\zs(u)&=&\pm\imath_{a\ket{0}+b\ket{1}}u_\pm=\\&&\pm(a\bra{0}+b\bra{1})\ket{0}\ket{1}
\pm(a\bra{0}+b\bra{1})\ket{1}\ket{0}=\pm a\ket{1}\pm b\ket{0}\,.\eeas
\end{example}

\begin{example} For the GHZ-states $\ket{GHZ_k}$ and W-states $\ket{W_k}$ the S-rank is 2 independently on $k\ge 2$. Indeed, it is clear that contractions of
$$\ket{GHZ_k}=\frac{1}{\sqrt{2}}\left(\ket{0}^{\ot k}+\ket{1}^{\ot k}\right)$$
with $(k-1)$-tensors give all linear combinations of $\ket{0}$ and $\ket{1}$. The same is true for
$$\ket{W_k}=\frac{1}{\sqrt{k}}\left(\ket{0\cdots01}+\ket{0\cdots10}+\cdots+\ket{1\cdots00}\right)\,.$$
On the other hand, the S-rank cannot exceed 2 for qubit systems.

\begin{remark} The S-rank does not distinguish between $\ket{GHZ_k}$ and $\ket{W_k}$. Note however that we can slightly generalize our notion of the S-rank including also contractions $\zi_\zn u$ with shorter tensors, i.e. tensors $\zn\in\cH^{\ot(k-l)}$ with $0<l<k$.
We have not insist on this generalization in order to avoid additional technical complications. The simple version of the S-rank is sufficient for distinguishing simple tensors. This extended version could be useful in measuring the entanglement.
\end{remark}

\end{example}

\begin{example} If $\ket{i_1}$ and $\ket{i_2}$ are orthonormal sets in $\cH_1$ and $\cH_2$, then the S-rank of $u\sum_{i=1}^r\zl_i\ket{i_1}\ot\ket{i_2}$ is $r$, as $\zi_{\cH_1}u$ is spanned by $\ket{i_2}$, $i=1,\dots,r$, and $\zi_{\cH_2}\zs(u)$ for the transposition $\zs(f\ot g)=g\ot f$ is spanned by $\ket{i_1}$, $i=1,\dots,r$. In other words the S-rank equals the Schmidt rank in this case, so the S-rank is a natural generalization of the latter.
\end{example}

\section{Entanglement and Segre maps for Bose and Fermi statistics}
\label{sec:segre}

Similarly to the case of distinguishable particles (see
\cite{grabowski05,grabowski06}), the sets of all bosonic simple
pure states (resp., fermionic simple pure states) can be described as
the images of certain maps defined on the products of projective
Hilbert spaces, \emph{the generalized Segre maps}, as follows.

Consider first the standard Segre embedding $\operatorname{Seg}_k$
induced by the tensor product map:

\begin{equation}\label{segre}
\xymatrix  @C=7pt { {(\mathcal{H}_\circ)^{\times k}\ar[d]} &{\ni}&
{(x_1,\ldots,x_k)\ar@{|->}[d]}
 \ar@{|->}[rr] &&
  {\ x_1\otimes\cdots\otimes x_k} \ar@{|->}[d] &{\in} & \left(\mathcal{H}^{\otimes k}\right)_\circ
\ar[d]
 \\
{(\mathbb{P}\mathcal{H})^{\times k}}&{\ni}&
{(\rho_{x_1},\ldots,\rho_{x_k})}
\ar@{|->}[rr]^(.6)/-6pt/{\operatorname{Seg}_k} && {\
\rho_{x_1\otimes\ldots\otimes x_k}} &{\in} &
\mathbb{P}(\mathcal{H}^{\otimes k}) }
\end{equation}
where $\mathcal{H}_\circ=\mathcal{H}\setminus\{ 0\}$.

It is clear that the analogous map, $\operatorname{Seg}_k^\vee$,
for the Bose statistics should be
\begin{equation}\label{segre-boson}
\xymatrix  @C=7pt { {\mathcal{H}_\circ\ar[d]} &{\ni}&
{x\ar@{|->}[d]}
 \ar@{|->}[rrr] &&&
  {\ x^k} \ar@{|->}[d] &{\in} & \left(\mathcal{H}^{\vee k}\right)_\circ
\ar[d]
 \\
{\mathbb{P}\mathcal{H}}&{\ni}& {\rho_{x}}
\ar@{|->}[rrr]^(.6)/-6pt/{\operatorname{Seg}_k^\vee} &&& {\
\rho_{x^k}} &{\in} & \mathbb{P}(\mathcal{H}^{\vee k}) }
\end{equation}
where $x^k=x\vee\cdots\vee x=x\otimes\cdots\otimes x$ ($k$-factors), and for the
Fermi statistics:
\begin{equation}\label{segre-fermion}
\xymatrix  @C=7pt { {\mathcal{H}^{\times k}_\circ\ar[d]} &{\ni}&
{(x_1,\ldots,x_k)\ar@{|->}[d]}
 \ar@{|->}[rrr] &&&
  {\ x_1\wedge\cdots\wedge x_k} \ar@{|->}[d] &{\in} & \left(\mathcal{H}^{\wedge k}\right)_\circ
\ar[d]
 \\
{(\mathbb{P}\mathcal{H})^{\times k}_\circ}&{\ni}&
{(\rho_{x_1},\ldots,\rho_{x_k})}
\ar@{|->}[rrr]^(.6)/-6pt/{\operatorname{Seg}_k^\wedge} &&& {\
\rho_{x_1\wedge\ldots\wedge x_k}} &{\in} &
\mathbb{P}(\mathcal{H}^{\wedge k}) }
\end{equation}
where $\mathcal{H}^{\times k}_\circ$ denotes
$$\mathcal{H}^{\times k}\setminus\{(x_1,\dots,x_k):x_1\wedge\cdots\wedge x_k=0\}$$
and $(\mathbb{P}\mathcal{H})^{\times k}_\circ$ is
$$(\mathbb{P}\mathcal{H})^{\times k}\setminus\{(\rho_{x_1},\dots,\rho_{x_k}):x_1\wedge\cdots\wedge x_k=0\}\,.$$
Note that the condition $x_1\wedge\cdots\wedge x_k\ne 0$ does not
depend on the choice of the vectors $x_1,\dots,x_k$ in their
projective classes and means that $\rho_{x_1},\dots,\rho_{x_k}$ do
not lie in a common projective hyperspace. The subset
$\mathcal{H}^{\times k}_\circ$ (resp.,
$(\mathbb{P}\mathcal{H})^{\times k}_\circ$) is open and dense in
$\mathcal{H}^{\times k}$ (resp., $(\mathbb{P}\mathcal{H})^{\times
k}$). The following is an immediate consequence of Theorem \ref{theo:simple}.
\begin{theo} A bosonic (fermionic) pure state $\rho\in \mathbb{P}(\mathcal{H}^{\vee k})$
(resp., $\rho\in \mathbb{P}(\mathcal{H}^{\we k})$) is entangled if
and only if it lies outside the range of the Segre map
$${\operatorname{Seg}_k^\vee}:{\mathbb{P}\mathcal{H}} \raa \mathbb{P}(\mathcal{H}^{\vee
k})\quad (\text{\rm resp.,}\quad
{\operatorname{Seg}_k^\we}:{(\mathbb{P}\mathcal{H})^{\times
k}_\circ} \raa \mathbb{P}(\mathcal{H}^{\we k})\, )\,.
$$
A mixed bosonic (fermionic) state is entangled if and only if it lies outside the convex hull of the range of the corresponding Segre map.
\end{theo}

\section{Entanglement for generalized parastatistics}
\label{sec:para}

Our approach to the entanglement of composite systems for
identical particles is so general and natural that it allows for
an immediate implications also for generalized parastatistics.
Parastatistics were introduced by Green \cite{green53} as a refinement of the spin-statistics connection introduced by Pauli \cite{pauli40}. Green was motivated by a two-pages paper by Wigner addressing the connection between equations of motion and the commutation relations (Wigner's problem) \cite{wigner50,manko97}. The context of Green's paper is Quantum Field Theory, while most of the applications which have been proposed deal with thermodynamical aspects, in particular with the calculation of the partition function.
In this paper we are concerned with parastatistics only to show that our proposed scheme applies to all the situations where states of composite systems are by construction not factorizable and are associated with representations of the permutation group acting on the tensor product of states of the subsystems. The reader who wants to know more about parastatistics could read \cite{ohnuki82,messiah64,cusson69}.

Observe first that simple tensors of length $1$ in
$\tilde{\mathcal{H}}=\mathcal{H}_1\otimes\cdots\otimes\mathcal{H}_k$
form an orbit of the group $U(\mathcal{H}_1)\times\cdots\times
U(\mathcal{H}_k)$ acting on $\tilde{\mathcal{H}}$ in the obvious
way. In fact, each such tensor can be written as
$e_1^1\otimes\cdots\otimes e_1^k$ for certain choice of
orthonormal bases $e_1^j,\dots,e_{n_j}^j$ in $\mathcal{H}_j$,
$j=1,\ldots,k$. This means that simple tensors are just vectors of
highest (or lowest, depending on the convention) weight of the
compact Lie group $U(\mathcal{H}_1)\times\cdots\times
U(\mathcal{H}_{k})$ relative to some choice of a maximal
torus and Borel subgroups. If indistinguishable particles are
concerned, the symmetric and antisymmetric tensors form particular irreducible parts of the
`diagonal' representation of the compact group $U(\mathcal{H})$ in
the Hilbert space $\mathcal{H}^{\otimes k}$, defined by
\begin{equation}\label{trep}
U(x_1\otimes\cdots\otimes x_k)=U(x_1)\otimes\cdots\otimes U(x_k).
\end{equation}
Recall that we identify the symmetry group $S_k$ with the group of
certain unitary operators on the Hilbert space $\mathcal{H}^k$ in
the obvious way,
$$\sigma(x_1\otimes\cdots\otimes x_k)=x_{\sigma(1)}\otimes\cdots\otimes x_{\sigma(k)}\,.$$
Note that the operators of $S_k$ intertwine the unitary action of
$U(\mathcal{H})$. In the cases of the symmetric and antisymmetric
tensors, we speak about the Bose and Fermi statistics, respectively.
But, for $k>2$, there are other irreducible parts of the
representation (\ref{trep}) associated with invariant subspaces
of the $S_k$-action. We shall call them {\it (generalized)
parastatistics}. Any of these $k$-parastatistics (i.e. any
irreducible subspace of the tensor product $\mathcal{H}^{\otimes
k}$) is associated with a {\it Young tableau} $\alpha$ with
$k$-boxes (chambers) as follows (see e.g. \cite{fulton91,fulton97, Zhel}).

Consider partitions of $k$: $k=\lambda_1+\cdots+\lambda_r$, where
$\lambda_1\ge\cdots\ge\lambda_r\ge 1$. To a partition
$\lambda=(\lambda_1,\dots,\lambda_r)$ is associated a {\it Young
diagram} (sometimes called a {\it Young frame} or a {\it Ferrers
diagram}) with $\lambda_i$ boxes in the $i$th row, the rows of
boxes lined up on the left. Define a {\it tableau} on a given
Young diagram to be a numbering of the boxes by the integers
$1,\dots,k$, and denote with $Y_\lambda$ the set of all such Young
tableaux. Finally, put $Y(k)$ to be the set of all Young tableaux
with $k$ boxes. Given a tableau $\alpha\in Y(k)$ define two
subgroups in the symmetry group $S_k$:
$$P=P_\alpha=\{\sigma\in S_k: \sigma \text{\ preserves each row of}\ \alpha\}$$
and
$$Q=Q_\alpha=\{\tau\in S_k: \tau \text{\ preserves each column of}\ \alpha\}\,.$$
In the space of linear operators on $\mathcal{H}^{\ot k}$ we
introduce two operators associated with these subgroups:
\begin{equation}\label{yoper}
a_\alpha=\sum_{\tau
\in Q}(-1)^{\tau}\tau\,,\quad b_\alpha=\sum_{\sigma \in P}\sigma\,.
\end{equation}
Finally, we define the {\it Young symmetrizer}
\begin{equation}\label{ysym}
c_\alpha=a_\alpha\circ b_\alpha=\sum_{\sigma \in P,\, \tau\in
Q}(-1)^{\tau}\tau\circ\sigma\,.
\end{equation}
It is well known (see, e.g., \cite{fulton91,fulton97,Zhel}) that $\pi^\alpha=\frac{1}{\mu(\alpha)}c_\alpha$,
for some non-zero rational number $\mu(\alpha)$, is an orthogonal
projector and that the image $\mathcal{H}^\alpha$ of $c_\alpha$ is
an irreducible subrepresentation of $U(\mathcal{H})$, i.e. the
{\it parastatistics associated with} $\alpha$. As a matter of fact,
these representations for Young tableaux on the same Young diagram
are equivalent, so that the constant $\mu(\alpha)$ depends only on
the Young diagram $\lambda$ of $\alpha$ and does not depend on the
enumeration of boxes. Hence, $\mu(\alpha)=\mu(\lambda)$ and this number is related
to the multiplicity $m(\lambda)$ of the corresponding irreducible
representation in $\mathcal{H}^{\ot k}$ by $\mu(\lambda)\cdot
m(\lambda)=k!$. For a given Young diagram (partition) $\lambda$,
the map
\begin{equation}\label{cysym}
\epsilon_\lambda=\frac{1}{\mu(\lambda)^2}\sum_{\alpha\in
Y_\lambda} c_\alpha
\end{equation}
is an orthogonal projection, called the {\it central Young
symmetrizer}, onto the invariant subspace being the sum of all
copies of the irreducible representations equivalent to that with
a parastatistics from $Y_\lambda$.

The symmetrization $\pi^\vee$ (antisymmetrization $\pi^\wedge$)
projection corresponds to a Young tableau with just one row (one
column) and arbitrary enumeration. It is well known that any
irreducible representation $\mathcal{H}^\alpha$ of
$U(\mathcal{H})$ contains cyclic vectors which are of highest
weight relative to some choice of a maximal torus and Borel
subgroups in $U(\mathcal{H})$. We will call them
\textit{$\alpha$-simple tensors} or \textit{simple tensors in
$\mathcal{H}^\alpha$}. Note that such vectors can be viewed as
\textit{generalized coherent states} \cite{perelomov86}. They can be
also regarded as the `most classical' states with respect to their correlation properties \cite{kb09}. These are exactly the tensors associated with simple
(non-entangled) pure states for composite systems of particles
with (generalized) parastatistics. This is because $\alpha$-simple
tensors represent the minimal amount of quantum correlations for
tensors in $\mathcal{H}^\alpha$, namely the quantum correlations forced
directly by the particular parastatistics.

\begin{example}\label{ex}
\begin{description}
\item{(a)} For $k=2$ we have just the obvious splitting of
    $\mathcal{H}^{\otimes 2}$ into symmetric and antisymmetric tensors:
    $\mathcal{H}^{\wedge 2}\oplus\mathcal{H}^{\vee 2}$.
\item{(b)} For $k=3$, besides symmetric and antisymmetric tensors associated with the Young tableaux
$$\Yvcentermath1 \alpha_0=\young(123) \quad \text{and} \quad
\alpha_3=\young(1,2,3)\,,$$
we have two additional irreducible parts associated
    with the Young tableaux
\end{description}
\begin{equation}\label{}
\Yvcentermath1 \alpha_1=\young(12,3) \quad and \quad
\alpha_2=\young(13,2)\ ,
\end{equation}
namely
\begin{equation}\label{}
\mathcal{H}^{\otimes 3}=\mathcal{H}^{\wedge 3}\oplus
\mathcal{H}^{\alpha_1}\oplus\mathcal{H}^{\alpha_2}\oplus
\mathcal{H}^{\vee 3}\,.
\end{equation}
Since $P_{\za_1}=\{ id, (1,2)\}$ and $Q_{\za_1}=\{ id, (1,3)\}$, we get
$a_{\za_1}=id-(1,3)$ and $b_{\za_1}=id+(1,2)$, so that
$$c_{\za_1}=a_{\za_1}\circ b_{\za_1}=\left(id-(1,3)\right)\circ\left(id+(1,2)\right)
=id+(1,2)-(1,3)-(123)\,.$$
As the multiplicity of the representation is 2, we have $\zm(\za_1)=3!/2=3$ and the projection
\begin{equation}\label{}
\pi^{\alpha_1}:\mathcal{H}^{\otimes
3}\rightarrow\mathcal{H}^{\alpha_1},
\end{equation}
takes the form
\begin{equation}
\pi^{\alpha_1}(x_1\otimes x_2\otimes x_3)=\frac{1}{3} (x_1\otimes
x_2\otimes x_3 +x_2\otimes x_1\otimes x_3 -x_3\otimes x_2\otimes
x_1 -x_3\otimes x_1\otimes x_2)\,.
\end{equation}
Similarly,
\begin{equation}\label{}
\pi^{\alpha_2}:\mathcal{H}^{\otimes
3}\rightarrow\mathcal{H}^{\alpha_2},
\end{equation}
\begin{equation}
\pi^{\alpha_2}(x_1\otimes x_2\otimes x_3)=\frac{1}{3} (x_1\otimes
x_2\otimes x_3 +x_3\otimes x_2\otimes x_1 -x_2\otimes x_1\otimes
x_3 -x_2\otimes x_3\otimes x_1).
\end{equation}
\end{example}
The simple tensors (the highest weight vectors) in
$\mathcal{H}^{\alpha_1}$ can be written as
\begin{equation}\label{}
v^{\alpha_1}_\lambda=\lambda(e_1\otimes e_1\otimes e_2-e_2\otimes
e_1 \otimes e_1),
\end{equation}
for certain choice of an orthonormal basis $e_i$ in $\mathcal{H}$
and $\lambda\ne 0$. Analogously, the simple tensors in
$\mathcal{H}^{\alpha_2}$, in turn, take the form
\begin{equation}\label{}
v^{\alpha_2}_\lambda=\lambda(e_1\otimes e_2\otimes e_1-e_2\otimes
e_1 \otimes e_1).
\end{equation}
For $\text{dim}(\mathcal{H})=3$, the simple tensors of length $1$
form an orbit of the unitary group $U(\mathcal{H})$ of the (real)
dimension $7$ in $\mathcal{H}^{\alpha_1}$ and
$\mathcal{H}^{\alpha_2}$. The simple symmetric tensors of length
$1$ form an orbit of the dimension $5$, and the simple
antisymmetric ones (of length $1$) -- an orbit of the dimension 1.
The dimensions of the irreducible representations are:
$\text{dim}(\mathcal{H}^{\we 3})=1$,\
$\text{dim}(\mathcal{H}^{\vee 3})=10$,\
$\text{dim}(\mathcal{H}^{\alpha_1})=\text{dim}(\mathcal{H}^{\alpha_2})=8$\,.

\medskip
A fundamental observation is that $\alpha$-simple tensors can also be characterized in terms of the S-rank.
\begin{theo}\label{crucial} A tensor $v\in\cH^\alpha$, $\alpha\in Y(k)$, is simple if and only if it has the minimal S-rank among non-zero tensors from $\cH^\alpha$. This minimal S-rank equals the number $r$ of rows
in the corresponding Young diagram and the simple tensor reads as
\be\label{a-simple} v=\pi^\alpha\left(e_{\alpha(1)}\otimes\cdots\otimes
e_{\alpha(k)}\right)\,,
\ee
where $e_1,\dots,e_r$ are some linearly independent vectors in $\cH$ and $\alpha(i)$ is the number of the row in which the box with the number $i$ appears in the tableaux
$\alpha$. In other words, the tensor
\be\label{a-simple1}
E_{\alpha}=e_{\alpha(1)}\otimes\cdots\otimes e_{\alpha(k)}
\ee
is the tensor product of $k$ vectors
from the sequence $E=(e_1,\dots,e_r)$ obtained by putting $e_j$ in the places indicated
by the number of the boxes in the $j$th row.
\end{theo}
\begin{proof} Assume that $v$ is of the form (\ref{a-simple}). Passing to the complexification $GL(\cH)$ of $U(\cH)$ and using a basis $e_1,\dots,e_n$ extending the linear independent family $e_1,\dots,e_r\in\cH$ to identify $GL(\cH)$ with $GL(n;\C)$, we easily see that $E_\alpha$, thus $v$, is an eigenvector for any diagonal matrix and is killed by any upper-triangular matrix. Moreover, $v\ne 0$, hence $v$ is a vector of highest weight. Indeed, by definition $b_\alpha(E_\alpha)$ is non-zero and proportional to $E_\alpha$, so that $v=t\cdot a_\alpha(E_\alpha)$ for a non-zero constant $t$. It is now easy to see that $\{\tau(E_\alpha):\tau\in Q\}$ is a family of linearly independent tensors in $\cH^{\ot k}$, so that $v=\sum_{\tau
\in Q}(-1)^{\tau}\tau(E_\alpha)$ is non-zero. As tensors from $\{\tau(E_\alpha):\tau\in Q\}$ are linearly independent and have one of $e_1,\dots,e_r$ as the last factor, the S-rank of $v$ is at least $r$. On the other hand, $v$ is composed from tensor products of $e_1,\dots,e_r$ only, so the S-rank is at most $r$.

Conversely, let $v\in\cH^\alpha$, $v\ne 0$. Without loss of generality we can assume that the numbers in the first column of the Young tableaux $\alpha$ are $k-r+1,\dots, k$. Since $\epsilon_\lambda(v)=v$, we have
$$v=\frac{1}{\mu(\lambda)^2}\sum_{\varsigma\in S_r}(-1)^\varsigma\varsigma\left(\sum_{\tau\in Q'}(-1)^\tau\tau(b_\alpha(v))\right)\,,$$
where
$$Q'=\{\tau\in Q: \tau \text{\ is identical on the first column of}\ \alpha\}$$
and $S_r$ is the permutation group of $\{ k-r+1,\dots, k\}$. This means that $v$ is skew-symmetric with respect to the last $r$ positions, so the contractions $i_\nu v$, with $\nu\in(\cH)^{\otimes(k-1)}$, can be written as contractions of a skew-symmetric $r$-tensor, thus they span a vector space of dimension $\ge r$ (see Theorem \ref{theo:simple} (c)). If this dimension is exactly $r$, then $v$ can be written as a combination of linearly independent tensor products $e_{s(1)}\otimes\cdots\otimes e_{s(k)}$ of vectors from $\{e_1,\dots,e_r\}$.
Since the tensor is symmetric with respect to permutations preserving rows and skew-symmetric with respect to permutations of the first column, each tensor product $e_{s(1)}\otimes\cdots\otimes e_{s(k)}$ in this combination should satisfy
$e_{s(i)}=e_{s_(j)}$ if $i$ and $j$ are in the same row of $\alpha$. Hence, $v$ is proportional to $\pi^\alpha(E_\alpha)$.

\end{proof}

Let $\mathcal{H}^\alpha\subset\mathcal{H}^{\otimes k}$ denotes the
irreducible component of the tensor representation of the unitary
group $U(\mathcal{H})$ in $\mathcal{H}^{\otimes k}$ associated
with a Young diagram $\alpha\in Y(k)$.
\begin{definition}\

\begin{description}
\item{(a)} We say that a pure state $\rho_v$ on $\mathcal{H}^{\otimes k}$ {\it
obeys a parastatistics} $\alpha\in Y(k)$ (is a {\it pure
$\alpha$-state} for short) if $v\in\mathcal{H}^\alpha$, i.e. $\rho$ is a pure state on the Hilbert space
$\cH^\alpha$.
\item{(b)} A pure state $\rho$ on $\mathcal{H}^{\otimes k}$
obeying a parastatistics $\alpha$ is called a \emph{simple pure
state for the parastatistics} $\alpha$ (\emph{simple pure
$\alpha$-state} for short) if $\rho$ is represented by an
$\alpha$-simple tensor in $\mathcal{H}^\alpha$. If $\rho$ is not
simple $\alpha$-state, we call it an \emph{entangled pure
$\alpha$-state}.
\item{(c)} A mixed state $\rho$ on $\mathcal{H}^{\alpha}$ we call
a \emph{simple (mixed) state for the parastatistics} $\alpha$
(\emph{simple $\alpha$-state} for short), if it can be written as a
convex combination of  simple pure $\alpha$-states. In the other
case, $\rho$ is called an \emph{entangled mixed $\alpha$-state}.
\end{description}
\end{definition}

\section{Segre maps for generalized parastatistics} \label{sec:segrepara}

In general, for an arbitrary parastatistics (Young tableau)
$\alpha\in Y(k)$ with the partition (Young diagram)
$\lambda=(\lambda_1,\dots,\lambda_r)$, we define the generalized
Segre map $\operatorname{Seg}^\alpha$ ({\it $\alpha$-Segre map})
as a map
$\operatorname{Seg}^\alpha:(\mathbb{P}\mathcal{H})^{\times
r}_\circ\rightarrow \mathbb{P}(\mathcal{H}^\alpha)$ described as
follows.

Let us consider first the map
$$i_\alpha:\mathcal{H}^{\times r}\rightarrow \mathcal{H}^{\otimes k}\,,
\quad (x_1,\dots,x_r)\mapsto x_{\alpha(1)}\otimes\cdots\otimes
x_{\alpha(k)}\,,$$ where $\alpha(i)$ is the number of the row in
which the box with the number $i$ appears in the tableaux
$\alpha$. In other words, we make a tensor product of $k$ vectors
from $\{ x_1,\dots,x_r\}$ by putting $x_j$ in the places indicated
by the number of the boxes in the $j$th row. For instance, the
Young tableaux from Example \ref{ex} give
$i_{\alpha_1}(x_1,x_2)=x_1\otimes x_1\otimes x_2$ and
$i_{\alpha_2}(x_1,x_2)=x_1\otimes x_2\otimes x_1$. It is clear
that $i_\alpha(x_1,\dots,x_r)$ is an eigenvector of $b_\alpha$.

The Segre map $\operatorname{Seg}^\alpha$ associates with
$(\rho_{x_1},\ldots,\rho_{x_r})\in(\mathbb{P}\mathcal{H})^{\times
r}_\circ$ the pure state
$\rho_{\pi^\alpha(x_{\alpha(1)}\otimes\cdots\otimes
x_{\alpha(k)})}$ in $\mathcal{H}^{\alpha}$, as shows the following
diagram:
\begin{equation}\label{segre-general}
\xymatrix  @C=7pt { {\mathcal{H}^{\times r}_\circ\ar[d]} &{\ni}&
(x_1,\ldots,x_r)\ar@{|->}[d]\ar@{|->}[rrr]^(.6)/-25pt/{\pi^\alpha\circ
i_\alpha}
  &&&
    \pi^\alpha(x_{\alpha(1)}\otimes\cdots\otimes x_{\alpha(k)})\ar@{|->}[d]&{\in}&\mathcal{H}^\alpha_0
\ar[d]
 \\
{(\mathbb{P}\mathcal{H})^{\times r}_\circ}&{\ni}&
(\rho_{x_1},\ldots,\rho_{x_r})
\ar@{|->}[rrr]^(.6)/-15pt/{\operatorname{Seg}^\alpha} &&&
\rho_{\pi^\alpha(x_{\alpha(1)}\otimes\cdots\otimes x_{\alpha(k)})}
&{\in} & \mathbb{P}(\mathcal{H}^{\alpha}) }
\end{equation}
Note that $\pi^\alpha(x_{\alpha(1)}\otimes\cdots\otimes
x_{\alpha(k)})$ is proportional to the antisymmetrization of the
tensor $x_{\alpha(1)}\otimes\cdots\otimes x_{\alpha(k)}$ and that
the construction is correct, since
$\pi^\alpha(x_{\alpha(1)}\otimes\cdots\otimes x_{\alpha(k)})$ is
non-zero if and only if $x_1\wedge\dots\wedge x_r\ne 0$, and its
projective class is uniquely determined by the projective classes
of $x_1,\dots,x_r$. Note also that we can always take
$x_1,\dots,x_r$ orthogonal, say, $x_1=e_1,\dots,x_r=e_r$, since
the antisymmetrization kills the part of $x_i$ which is the
orthogonal projection of $x_i$ onto the linear subspace spanned by
the rest of the vectors $x_j$. Now, according to Theorem \ref{crucial},
$\pi^\alpha(x_{\alpha(1)}\otimes\cdots\otimes x_{\alpha(k)})$ is
$\alpha$-simple, so that
$\rho_{\pi^\alpha(x_{\alpha(1)}\otimes\cdots\otimes
x_{\alpha(k)})}$ is a simple pure $\alpha$-state.
Moreover, each simple pure $\alpha$-state is of this form and for symmetric and antisymmetric tensors this
construction agrees with (\ref{segre-boson}) and
(\ref{segre-fermion}). We therefore get the following.
\begin{theo} A pure $\alpha$-state $\rho\in \mathbb{P}(\mathcal{H}^\alpha)$
is an entangled $\alpha$-state if and only if it is not in the
range of the Segre map
$$\operatorname{Seg}^\alpha:(\mathbb{P}\mathcal{H})^{\times r}_\circ \raa
\mathbb{P}(\mathcal{H}^\alpha)\,.
$$
The set $\cS(\cH^\alpha)$ of mixed non-entangled $\alpha$-states
is the convex hull of the range of the normalized
$\alpha$-Segre map,
$$\cS(\cH^\alpha)=conv\left({\operatorname{Seg}^\alpha}\left(\cP(\cH)^{\times
r}_\alpha\right)\right)\,,
$$
and mixed entangled $\alpha$-states are exactly members of
$$\cD(\cH^\alpha)\setminus \cS(\cH^\alpha)\,.
$$
\end{theo}

Let us observe that
$$\Vert\pi^\alpha(x_{\alpha(1)}\otimes\cdots\otimes
x_{\alpha(k)})\Vert^2\cdot\rho_{\pi^\alpha(x_{\alpha(1)}\otimes\cdots\otimes
x_{\alpha(k)})}=\pi^\alpha\circ\left(\rho_{x_{\alpha(1)}}\otimes\cdots\otimes
\rho_{x_{\alpha(k)}}\right)\circ\pi^\alpha\,,
$$ as operators on $\cH^{\otimes k}$. Indeed,
since $\pi^\alpha$ is an orthogonal projection, the left-hand side
equals
\begin{eqnarray*}
&\kb{\pi^\alpha(x_{\alpha(1)}\otimes\cdots\otimes
x_{\alpha(k)})}{\pi^\alpha(x_{\alpha(1)}\otimes\cdots\otimes
x_{\alpha(k)})}(y)=\\
&\pi^\alpha(x_{\alpha(1)}\otimes\cdots\otimes
x_{\alpha(k)})\bk{x_{\alpha(1)}\otimes\cdots\otimes
x_{\alpha(k)}}{\pi_\alpha(y)}=\left(\pi^\alpha\circ\left(\rho_{x_{\alpha(1)}}\otimes\cdots\otimes
\rho_{x_{\alpha(k)}}\right)\circ\pi^\alpha\right)(y)\,.
\end{eqnarray*}
This suggests to look at the map (the {\it big $\alpha$-Segre
map}):
$$\wt{\operatorname{Seg}^\alpha}:(\mathfrak{u}^*(\cH))^{\times
r}\raa \mathfrak{u}^*(\cH^\alpha)\subset
\mathfrak{u}^*(\cH^{\otimes k})\,,
$$
where $\mathfrak{u}^*(\cH)$ denotes the (real) vector space of
selfadjoint operators on $\cH$, defined by
$$\wt{\operatorname{Seg}^\alpha}(u_1,\dots,u_r)=\pi_\alpha\circ(u_{\alpha(1)}\otimes\cdots\otimes
u_{\alpha(k)})\circ\pi_\alpha\,.
$$


The big $\alpha$-Segre map is a natural generalization
of the map (\ref{imb}). A closer study of the big Segre maps we postpone to a separate paper.

\section{Conclusions}
The presented geometric description, in terms of Segre maps, of entanglement properties for systems with arbitrary statistics parallels our previous algebraic approach to such systems based on the concept of S-rank of a tensor. It puts on equal and unifying footing systems of distinguishable particles, for which both algebraic and geometric descriptions were known, and systems with indistinguishable particles. What is more, this description provides an explicit form of simple (separable) pure states for an arbitrary parastatistics an effective procedures to check the simplicity in the bosonic and the fermionic case. Such procedures for arbitrary parastatistics are not known for us and can be the subject of forthcoming papers. Also the problem of the decomposition of the total algebra of operators, associated with the S-rank, is an interesting and open question.

\section*{Acknowledgments}
Research of the first two authors was financed by the Polish Ministry of Science and Higher Education under the grant No. N N202 090239. G.~Marmo would like to acknowledge the support provided by the Santander/UCIIIM chair of Excellence programme 2011-2012. We are also indebted to referees for their useful comments and suggestions.


\end{document}